\documentclass{jloganal}    

\title{Left c.e.\ reals $\alpha$ and $C(C(\alpha_n)|\alpha_n)$.}

\author{George Davie}
\givenname{George}
\surname{Davie}
\address{Department of Decision Sciences, School of Economic and Financial Sciences, University of South Africa, South-Africa}
\email{davieg@unisa.ac.za}
\urladdr{}


\arxivreference{}
\arxivpassword{}

\volumenumber{}
\issuenumber{}
\publicationyear{}
\papernumber{}
\startpage{}
\endpage{}
\doi{}
\MR{}
\Zbl{}
\received{}
\revised{}
\accepted{}
\published{}
\publishedonline{}
\proposed{}
\seconded{}
\corresponding{}
\editor{}
\version{}

\newtheorem{thm}{Theorem}[section]    
\newtheorem{lem}[thm]{Lemma}          
\newtheorem{prop}[thm]{Proposition}    
\newtheorem{cor}[thm]{Corollary}          
\newtheorem*{mlem}{Main Lemma}   
\newtheorem{ex}[thm]{Example}    

\newtheorem*{chait}{Chaitin}

\newtheorem*{Thm1}{Theorem 1} 

\newtheorem*{Cor2}{Corollary 2} 
\newtheorem*{Thm2}{Theorem 2}
\newtheorem*{Thm3}{Theorem 3} 
 
\theoremstyle{definition}
\newtheorem{defn}[thm]{Definition}    
\newtheorem*{rem}{Remark}             
\begin{document}

\begin{abstract}
We are interested in the computability between left c.e.\ reals $\alpha$ and their initial segments. We show that the quantity $C(C(\alpha_n)|\alpha_n)$ plays a crucial role in this and in their completeness. We look in particular at Chaitin's theorem and its relativisation due to Frank Stephan.
\end{abstract}

\maketitle


\date{June 2022}

\section{Introduction}
Let $\alpha$ and $\beta$ be a left c.e.\ reals and $\alpha_n$, $\beta_n$ their first $n$ digits. It is obvious that $\alpha_n$ computes $\beta_n$ if it's settling time -- the time it takes for the c.e.\ real to settle on $\alpha_n$ -- is longer than that of $\beta_n$. This will be a consequence if the Kolmogorov complexity  $C(\alpha_n)$ is sufficiently greater than  $C(\beta_n)$. 

How hard is it to compute such a $\alpha_n$ from $\beta_n$? Our basic observation is that there exists a smallish $d$\footnote{$d$ computable in $l(P_{\beta},l(P_{\alpha})$}  such that, if $C(\alpha_n) \geq C(\beta_n)+d$ then it is often very hard. 

We will show that this fact implies sharp results around Chaitin's characterisation of computability in terms of initial segment complexity. For example, we will show that if $C(\alpha_n) \geq C(n)+d$ for all $n$, then $\alpha$ is Turing-complete in a very strong sense. We will also look at Frank Stephan's relativisation of Chaitin's result, in this light.

More generally, we will examine the difficulty of computing an initial segment $\alpha_n$ of even slightly higher (or lower)  complexity than that of a given $\beta_n$.

 \section{Background and definitions}

\subsection{Kolmogorov complexity}
\begin{defn}
Let $U$ be a universal and additively optimal binary Turing machine. Let $x$ and $y$ be binary strings. The
Kolmogorov complexity of $x$, given $y$, is the length $l(P)$ of a shortest program $P$ for $U$ that outputs $x$ on input $\langle y,P \rangle $. This is denoted by $C(x|y)$.
The unconditional complexity is then defined as $C(x)=C(x|\epsilon)$ where $\epsilon$ denotes the empty string.
\end{defn}
For $C(x|y)$, imagine $y$ to be given for free, for the program to use as it wishes. For example, if $y=x$ then $C(x|y)$ is a constant: ``write the given $y$''. Similar if $x$ is the reverse of $y$ and so on. 

\smallskip

We will generally follow the notations and conventions of the standard reference for Kolmogorov complexity, Li and Vit\'{a}nyi \cite{lv}. (See in particular Chapters 2 and 3.) The reader is referred here for background and 
undefined concepts and notation. 

\subsubsection{$C(C(x)|x)$}
Most strings of length $n$ have Kolmogorov complexity $C(x)$ close to $n$, and so readily give their own complexity: ``my length minus something small''.  Some strings, however, do \emph{not} want to give their complexity. 
What does that mean, carefully? 
\\It means that there is no short program, which takes $x$, and outputs $C(x)$.

In other words, $C(C(x)|x)$ is large.    

Before we continue, we note that there \emph{are} strings $x$ with large $C(C(x)|x)$. In fact, see Shen and Bauwens \cite{Shen}\footnote{Also see this paper for references to previous work on the function $C(C(x)|x)$.}: For each $n$ there are strings of length $n$ such that
$$C(C(x)|x)\geq \log n-O(1)\footnote{Recall that, since the length $C(C(x)|x))$ of such a program  for $x$ of length $n$ is bounded from above by $\log n$, see Li and Vit\'{a}nyi \cite{lv}, this is the best possible bound.}.$$   For such strings $x$, $x$ itself does not help at all in finding $C(x)$.
\\Strings $x$ with large $C(C(x)|x)$ are rare but play a central role in the theory. One example is in counterexamples to information symmetry of Kolmogorov complexity, see G\'{a}cs, \cite{Gacs}.
\\ Also see, for example, Section 2.8 in Li and Vit\'{a}nyi \cite{lv}. We will sometimes call strings for which $C(C(x)|x)$ is not small, \emph{complexity withholding}.

\subsection{Left c.e.\ reals}
A left c.e.\ real is a limit of a computable sequence of rationals. We can require that each element of the sequence is itself computable without changing the class, see Downey and Hirschfeldt  \cite[Theorem 5.1.5]{Down}. Note that one can consider c.e.\ \emph{sets} $S$ as a subset of this class -- put a 1 in position $n$ iff $n\in S$.
By an abuse of notation we use $P_{\alpha}$ for a program for the enumeration of the sequence of rationals converging to the infinite c.e. real $\alpha$.

\begin{ex}
Dovetail the running of all programs on the empty input, whenever one halts, say program $P$, add measure $2^{-l(P)}$ to the sum. This converges to the best known left c.e.\ real of all, Chaitin's $\Omega$, see \cite{Chait}.
\end{ex}
We denote a generic left c.e.\ real by $\omega$ and the first $n$ bits by $\omega_n$. We will use $\alpha$ and $\beta$ for two particular left c.e.\ reals.

We note the following definitions:
\begin{defn}\label{Tred}
	$\alpha \leq _{T}\beta $ iff $\alpha$ is  Turing-reducible to $\beta$, i.e.\ $\alpha$ is computable from $\beta$.
\end{defn}
When we call a left c.e.\ real \emph{complete}, we will always mean Turing-complete.
\begin{defn}\label{cred}
	$\alpha \leq _{C}\beta $ iff $|C(\alpha_n)- C(\beta_n)|=O(1)$ for all $n$.
\end{defn}
That is, complexities of initial segments are always within a constant. 
\smallskip
\\This definition seems much less complicated than it seems, which will come out implicitly in the paper, for example, it matters a lot whether the inner term is positive or negative. Positive enough implies $\beta_n$ easily computable in $\alpha_n$, negative enough implies $\alpha_n$ easily computable in $\beta_n$\footnote{Via settling times.}. But the inner term may change size and sign. 
 In fact, it needs a remarkable proof (Stephan) just to show that it implies mutual computability $\alpha \leq _{T}\beta $ and $\beta \leq _{T}\alpha $, Theorem \ref{rel chait Steph}.

\begin{defn}\label{rcred}
	$\alpha \leq _{rC}\beta $ iff $C(\alpha_n|\beta_n)=O(1)$ for all $n$.
\end{defn}

That is, there is some constant, such that for each $n$, we can transform\footnote{We will use ``transform $a$ to $b$'' as meaning ``compute $b$ from $a$''.}$\beta_n$ to $\alpha_n$ with a program of length bounded by the constant. It can be shown, see Downey and Hirschfeldt \cite[Theorem 9.6.8]{Down}, that:	

\begin{thm}(Downey, Hirschfeldt and Laforte)\label{DHL}
	$C(\alpha_n|\beta_n)=O(1)$ for all $n$ implies $\alpha \leq _{T}\beta$.
\end{thm} 
An ordering closer in spirit to our theme is that \emph{Settling time reducible}, of Csima and Shore, see \cite{Csima}.
\\As standard references for left c.e.\ reals and complexity, the reader is referred to Nies, \cite{Nies}, or Downey and Hirschfeldt, \cite{Down}.
	
\section{Main idea, Main lemma and theorems}
The single theme running through the paper is the following: 
\\For initial segments $\beta_n$,$\alpha_n$ of left c.e.\ reals, there exist a $d \in N$, computable in $(l(P_{\alpha}),l(P_{\beta}))$, such that, if $C(\alpha_n)-C(\beta_n) \geq d$, then $\alpha_n$ computes $C(\beta_n)$ given $\beta_n$. Almost all the results follow from this.

\subsection{Two processes}\label{dovetail}
Two standard processes play the central role in our results:
	\begin{itemize}
		\item Dovetailing of programs on the empty input: Run program 1 for step 1, then program 2 for step 1 and program 1 for step 2, then \ldots.
		 As the programs $P_i$ halt, list the pair $(P_i,x)$ where $x$ is the output of $P_i$.
		\item For a left c.e.\ real $\alpha$ : Run a program $P_{\alpha}$ which enumerates the binary rationals which are converging to $\alpha$. Note that in general, it will take a large number of steps for the first $n$ digits to settle on the (final) first $n$ digits $\alpha_n.$ The time it takes to settle on the segment $\alpha_n$ on the first $n$ digits is called the \emph{settling time} of $\alpha_n$.
	\end{itemize}

\subsection{Main lemma}
Our main tool is the following elementary lemma.	
	\begin{mlem}\label{mainlem}
		An initial segment $\alpha_n$  gives the complexities of all strings $x$ of complexity up to $C(\alpha_n)-d$, with $d$ computable in $l(P_{\alpha})$.
	\end{mlem}
\begin{proof}All the pairs $(x,P_x^{*})$ appear in the program dovetailing \emph{before} the first $n$ digits settle as $\alpha_n$: 	
If such a pair $(x,P_x^{*})$ appears \emph{after} $\alpha_n$ settles, then the program $P_x^{*}$ (of complexity less than that of $\alpha_n$) gives $\alpha_n$ as follows: 
Run the program dovetailing until $(P_x^{*},x)$ appears, now read off the current first $n$ digits of the c.e.\ sequence process above - this will be $\alpha_n$.  
\end{proof}

To use $P_x^{*}$ to find $C(\alpha_n)$, given $\alpha_n$, we only need any program $P_{\alpha}$ and a program for the dovetailing. It is thus clear that $d$ it is computable in $l(P_{\alpha})$ and the normal dovetailing overheads. 
\\In what follows $x$ itself will often be a $\beta_n$. In this case, $d$ will be computable in $l(P_{\alpha})$ and $l(P_{\beta})$.
\begin{lem}
		An initial segment $\alpha_n$  gives the complexities of all initial segments $\beta_n$ of complexity up to $C(\alpha_n)-d$, with $d$ computable in $l(P_{\alpha})$ and $l(P_{\beta})$.
\end{lem}	
Note that such an $\alpha_n$ does not merely have longer settling time than $\beta_n$. It's settling time is so much longer, that $\beta_n$ and $P^{*}(\beta_n)$ appear in the program dovetailing before $\alpha_n$ has settled.\footnote{In fact, we could always tweak $\beta_n$ to an $\alpha_n$ by making sure we always add something to $\alpha_n$ after something has added to $\beta_n$. Such an $\alpha_n$ settles after $\beta_n$ but does not give $C(\beta_n)$.}

\subsubsection{Computable}

	We have the famous result of Chaitin:
	\begin{thm}\label{chait orig}
	$\alpha_n$ is  computable  if and only if there is a $d\in N$ for which $$C(\alpha_n)\leq C(n)+d$$ for all $n$.
	\end{thm}
Equivalently
	\begin{chait}\label{chait contra}
		$\alpha_n$ is a non-computable left c.e.\ real if and only if, for each $d\in N$;  $$C(\alpha_n)\geq C(n)+d$$ for some $n$.
	\end{chait}
We can sharpen this in an interesting way
\begin{thm}
	Every noncomputable sequence $\alpha$ computes a partial function $C^{*}(n)$ on inputs $n$, for which infinitely often $C^{*}(n)=C(n)$. Hence, infinitely often, $\alpha$ is correct about $C(n)$. This happens at all the points at which $C(\alpha_n)>C(n)+d$.
\end{thm}
\begin{proof}
	For each constant $c$, $C(\alpha_n)>C(n)+c$ infinitely often, including for $c=d$. Using the Main lemma, by the time $\alpha_n$ with $C(\alpha_n)>C(n)+d$ has settled, all programs of  length at most $C(n)$ have halted with their shortest programs. Hence, given any $n$, find the pair $(n,p)$ with shortest $p$, then $|p|=C(n)$.
\end{proof}	
In some sense, all non-computable $\alpha$ are a bit complete.
\\If $C(\alpha_n)>C(n)+d$ for all $n$, then
	\begin{Thm1}
		Given $\alpha$, there is a computable $d\in N$ such that if $C(\alpha_n)>C(n)+d$ for all $n$, then $\alpha_n$  computes $C(0_n)$, hence $C(n)$, given $n$. 		
	\end{Thm1}
Which is a strengthening of the ``if'' part of Chaitin above, in terms of a particular computable $d$.

Hence $\alpha$ computes the complexity function and thus the halting problem. It follows that all such $\alpha$ are not only non-computable but in fact Turing-complete. 

\begin{cor}
Given $\alpha$, there is a computable $d\in N$ such that if $C(\alpha_n)>C(n)+d$ for all $n$, then $\alpha_n$  is Turing complete.
\end{cor}
\begin{rem}
	There is thus a tight bound $d$ such that $\alpha$ is computable if $C(\alpha_n)<C(n)+d_1$ for some $d_1<d$. At a bit larger $d>d_1$, $\alpha$ becomes complete. 
\end{rem}

 This leads to: 
\begin{cor}\label{wow}
	Let $\alpha $ be a Turing incomplete c.e.\ real. There is a computable constant $d$ such
	that $C(\alpha_n)<C(n)+d$ infinitely often. 
\end{cor}
\begin{rem}
The relation with existing results is quite subtle, Corollary \ref{wow} follows from Stephan's \emph{proof}, of Theorem \ref{rel chait Steph} in \cite[Theorem 9.7.1]{Down}. Having an explicit computable cut-off point $d$, for completeness, sharpens this result in one direction. The result that \emph{looks} the most similar \cite[Theorem 9.12.4 (Stephan)]{Down} has much stronger conditions, that $\lim C(\alpha_n)-C(n)\rightarrow \infty$.  
\end{rem}

The author finds it remarkable that all left c.e.\ reals except those which are Turing complete,\footnote{These can attain complexity larger than $n-2 \log n$ \emph{all the time}.} infinitely often have complexity within a constant of that of the $0$ sequence. 
That this is unexpected is emphasized by the fact that all such c.e.\ reals are among the \emph{paracomputable} sequences. See \cite[Exercise 2.5.15]{lv}. 

In fact, using the methods of Stephan's proof of Theorem 9.7.1 in \cite{Down}, we can show that sparse encodings of the halting sequence, can be ``stretched out'' such that they are also paracomputable.

\smallskip

Combining the two results:
\begin{cor}\label{complete1}
	Let $\alpha $ be a non-computable Turing incomplete c.e.\ real. Then 
	\\$C(\alpha_n)-C(n)$ is unbounded but there is also a computable constant $d$ such
	that 
	\\$C(\alpha_n)<C(n)+d$ infinitely often.
\end{cor}

\subsubsection{Relativising}
Recall Chaitin's theorem:
\begin{thm}\emph{[Chaitin]}
	A sequence $\alpha$ is computable if and only if there exists a $d$ such that $C(\alpha_n) \leq C(0_n)+d$ for all $n$.
\end{thm}
Stephan generalised the \emph{if} direction to the non-computable case, see \cite[Theorem 9.7.1]{Down} for the remarkable proof. See also Merkle and Stephan, \cite{Merkle Stephan}.

\begin{thm}\emph{[\textbf{Relative Chaitin} - Stephan]}\label{rel chait Steph}
	A left c.e.\ real $\alpha$ is computable in $\beta$ if there exists a $d$ such that $C(\alpha_n) \leq C(\beta_n) +O(1)$ for all $n$.
\end{thm}
In terms of the definitions in Section 2: 	$\alpha \leq _{C}\beta \rightarrow 	\alpha \leq _{T}\beta $.

\begin{rem}\label{halting sequence} There is of course no hope of relativising the other direction:  Take, for example, $\alpha$ as the set of halting programs as a c.e.\ sequence -- put a $1$ in position $n$ if the $n$-th program halts on input $n$ -- with $C(\alpha_n) \leq 2\log n+O(1)$. Also take $\Omega$ with $C(\Omega_n) \geq n- 2\log n +O(1)$. Then despite the vast complexity difference, each computes the other -- and both are complete. 
\end{rem}	

Relativising Theorem 1 for a left c.e.\ real $\beta$ instead of the $0$ sequence gives a sharp relativised form of the \emph{if} part of original form of Chaitin\footnote{This does not sharpen Theorem \ref{rel chait Steph} however, since Stephan does not assume that $C(\alpha_n) \geq C(\beta_n)$.}.  
\begin{Cor2}
Given $\alpha$,$\beta$, there is a computable $d\in N$ such that 
\\if $C(\alpha_n)>C(\beta_n)+d$ for all $n$, then $\alpha_n$  computes $C(\beta_n)$ on input $\beta_n$, for all $n$.
\end{Cor2}	
Note that ``$\alpha_n$ computes $C(\beta_n)$ on input $\beta_n$'' relativises ``$\alpha_n$ computes $C(0_n)$ on input $0_n$'' in Chaitin's theorem, which in the latter computable case implies Turing complete. 
\\Of course in the relativised case, we get Turing completeness too, but the result is stronger.
Note that, even if $\beta$ is itself complete, $\alpha_n$ computing $C(\beta_n)$ for each $n$ is quite strong.

\subsection{$\beta_n$ cannot lift or drop with $C_{\beta,n}$ resources, even if large }
\subsubsection{Lifting}
	\begin{Thm2}
		For a given  $\beta$,$\alpha$, there is a computable number $d$, such that, for all $n$, no program of length at most $C_{\beta,n}$ can transform $\beta_n$ into  $\alpha_n$ with $C(\alpha_n) \geq C(\beta_n)+d.$
	\end{Thm2} 

\begin{proof}
	We know that if we can compute such an $\alpha_n$ from $\beta_n$, then $\alpha_n$ can be used to form a list containing $\beta_n$ paired with a shortest program for $\beta_n $. Hence $C(\beta_n)$ can be found given $\beta_n$. We thus must have that $$C(\alpha_n|\beta_n) \geq C_{\beta,n}+d$$
\end{proof}
Hence, even when $C_{\beta,n}$ is very large, there is no program of length $C_{\beta,n} \gg d$ or less which can transform $\beta_n$ into an $\alpha_n$ of complexity even $d$ more.
\begin{rem}
	Hence we cannot try to modify $\beta_n$ and hope to reach some $\alpha_{n}$ which is $d$ above. For example, looking at how far other c.e.\ initial segments have developed and then modifying them with any programs of length at most $C_{\beta,n}$, cannot produce an actual c.e.\ segment which is more than $d$ above.
\end{rem}
\begin{cor}
	For a given  $\beta$,$\alpha$, there is a computable number $d$ such that, if 
	\begin{equation}\label{eq1}
	C(\alpha_n|\beta_n) \leq C_{\beta,n}+d
	\end{equation} for all $n$, then $\alpha$ is computable in $\beta$.
\end{cor} 
\begin{proof}
	We know that no program of length at most $C_{\beta,n}$ can lift $C(\beta_n)$ to a c.e.\ real even $d$ higher, hence equation (\ref{eq1}) actually means that $\alpha_{n}$ is within $d$ of $\beta_n$. By relativised Chaitin, \ref{rel chait Steph} we get that $\alpha$ is computable in $\beta$.  
\end{proof}
This weakens the $O(1)$ condition in Theorem \ref{DHL} to the expression $C_{\beta,n}$.
\subsubsection{Dropping}
Note that our interest is in $C_{\beta,n}$. Hence, we imagine we are given $\beta_n$ and are trying to find $C(\beta_n)$. In the theorem above we tried to \emph{lift} to an initial segment $\alpha_n$ for which the settling time was so long that $\beta_n$ and a shortest program for it, appeared. If, on the other hand, we could \emph{drop} to a close complexity $\alpha_n$ then we, as $\beta_n$, could play the settling time roll and wait for $\alpha_n$ and a shortest program for $\alpha_n$ to appear. Since we are assuming that $\alpha_n$ and $\beta_n$ have very close complexity, we could just add a small number to $C(\alpha_n)$ to find the sought after $C(\beta_n)$. The following theorem says that this is also hard. Note that in the the theorem above, we did not need to know exactly how much higher $\alpha_n$ was than $\beta_n$. For the case below, we must know how much lower:     
\begin{Thm3}
	For a given $\beta$,$\alpha$, there is a computable number $d$ such that, for all $n$, no program of length at most $C_{\beta,n}-2 \log c-d$ can transform $\beta_n$ into  $\alpha_n$ with $C(\alpha_n) = C(\beta_n)-c-d.$
	\end{Thm3}
\begin{proof}
	Use the Main lemma again. Firstly, find $\alpha_n$ from $\beta_n$. Now wait for $\beta_n$ to settle and look for $C(\alpha_n)$. To get $C(\beta_n)$, add $c$. 	
	Hence $$C_{\beta,n} \leq C(\alpha_n|\beta_n) +2\log c+d$$	
\end{proof}	
Here\ $c=C(\beta_n)-C(\alpha_n)$. Hence, even when $C_{\beta,n}$ is very large, there is no program of length $C_{\beta,n}-2\log c \gg d$ or less which can transform $\beta_n$ into an $\alpha_n$ of complexity even $d$ less.
\begin{rem}
There is a subtle difference between what makes lifting hard and what makes dropping hard. For lifting, it is hard to compute the higher $\alpha$ since we know such an $\alpha_n$ will supply us with shortest programs, including $(\beta_n,P^{*}(\beta_n))$, hence $C_{\beta,n}$. 
\\In the case of dropping, it is what we have, $\beta_n$, supplying the shortest programs. But now we need to know what $\alpha_n$ is, in order to find the correct pair $(\alpha_n,P^{*}(\alpha_n))$, hence $C(\alpha_n)$ hence $C(\beta_n)$ (add $c+d$). So again, finding $\alpha_n$ from $\beta_n$ must be difficult, but for slightly different reasons.
\end{rem}
Lifting is intuitively harder so it is perhaps even stranger that $\beta_n$ cannot be transformed into $\alpha_n$ of complexity at least $d$ less, unless the drop is large. That is, $C(\alpha_n)$ is \emph{much} smaller than $C(\beta_n)$, and the $2 \log c$ term dominates.

\section{Non-triviality -- The 0 sequence}	\label{0}
     It is of course a non-trivial statement that no program of length at most 
     $C_{\beta,n}$ can transform  $\beta_n$ into $\alpha_n$ only if $C_{\beta,n}$ is not itself close to $d$. Let us look at the example $0$.
 	 For most initial segments of the $0$ sequence, it is indeed the case that $C_{\beta,n}$ is close to $d$, since most $n$ are maximally complex or close to it. That is, the fraction of strings of length $n$ with $C_{0,n} \leq c$ grows rapidly with increasing $c$. However, as Shen and Bauwens \cite{Shen} have shown, there are, for each $n$, strings $x$ of length $n$ such that $C_x \geq \log n-O(1)$. Hence, in the $0$ sequence, there are segments of length $n$ which, even with access to any program of length at most $\log\log(n)$ cannot lift(drop) to a segment of complexity $d$ more(less). Since $\log\log(n) \rightarrow \infty$, this is remarkable.  
 	 
 	 In Shen and Bauwens \cite{Shen}, it is remarked (as part of the proof) that for each $k \in N$ there exists strings $x$ of length $n$ of complexity: $C(x) \geq \frac{k}{k+1}n-O(1)$  with $C_x \geq \log n-O(1)$. 
 	 
 	 Can left c.e.\ reals come close to having these two quantities so high? If so, then this would imply an extremely large ``reachability'' gap above and below. 
 	 
 \subsection{Close and unreachable, or none?}
 For the $0_n$ case above, and for general $\beta_{n}$, what does it mean that we cannot reach $\alpha_{n}$ of close complexity? There are two possibilities, there are such $\alpha_{n}$ but they are unreachable, or, there simply aren't any that are both close and at least $d$ away\footnote{Presumably we could taylor make such a left c.e.\ real, we mean when the $l(P_{\alpha})$ are negligible.}. 
 
 None of the results rule out that there are such initial segments close by. For example, the fact that we cannot lift to higher $\alpha_n$ does not mean they are not there. It does however mean that such $\alpha_n$ readily give their own complexity $C_{\alpha,n}$, by dropping to where we are working from, namely $\beta_n$.  
 
 On the other hand, if $\beta$ and $\alpha$ do not compute one another, then, by (the contrapositive of) Stephan's Theorem \ref{rel chait Steph}, the difference $|C(\alpha_n)-C(\beta_n)|$ is unbounded. Hence there is no reason that we should be able to find close $\alpha_n$. 
 \\If neither is complete, then yes, they must each be within $d$ of $C(n)$ infinitely often, by Corollary \ref{wow}, but this would happen at very different $n$.

\section{Complexity withholding strings}
According to a referee, the results in this section are known, we include them for interest and because we use Proposition \ref{cw string} below.

We move away from left c.e.\ reals and just look at normal strings. For this general case with $x$ any string, denote $C(C(x)|x)$ by $C_x$.
In the comment to Exercise 2.8.6 in Li and Vit\'{a}nyi \cite{lv}, on strings $x$ with high $C(C(x)|x)$, they state:
\begin{quote} This means that $x$ only marginally helps to compute $C(x)$; most information in $C(x)$ is extra information related to the halting problem.
\end{quote}
Also, in Exercise 2.8.1
\begin{quote}
The following equality and inequality seem to suggest that the shortest description of $x$ contains some extra information besides the description of $x$\footnote{We should add, only when $C_x$ is non-trivial.}.
\end{quote} 
In this section we will demonstrate this in a strong form. Roughly, that $C(x)$ solves the halting problem up to length $C_x-d$.
\begin{rem}
	Note that up to now, with $x=\alpha_n$  a segment of a left c.e.\ real,  $x$ implicitly contained a \emph{process} and with it, computational information -- it's own settling time -- \emph{itself}. That is, we could run $P_{\alpha}$ and in the time until $\alpha_n$ settles we will get shortest programs - although we will not know up to where if we do not have a lower bound for $C(\alpha_n)$. 
	\\For a general $x$, $x$ is not itself associated with a long process which gives computational information, it needs $C_x$.	
\end{rem}

Recall that, for most $x$, $C(x)$ is trivial in $x$, and $C(x)$ generally gives no information on the halting problem.

If $C_x$ is large, it implies that the time we need to wait for $x$ and a shortest program, $(x,P_x^{*})$ to appear in the enumeration of Section \ref{dovetail}, has high complexity, hence is long.
\begin{rem}
	For $C_x$ to be large, it also means that there are many programs which output $x$ before $P^{*}_x$ does. 
	So in some sense $x$ is disproportionally well represented amongst program outputs. 		 	
	This seems to hold in a stronger sense than just the fact that $C(x)$ being short allows for many ``padded up'' programs for $x$. This is because there must also be many programs shorter than a length $l$ for $l>C(x)$ and $C(l)$ small.
\end{rem}

\begin{prop}\label{cw string}
	There is a constant $d$, such that $(x,C(x))$ finds shortest programs, and hence finds $C(y)$, for all $y$ with $C(y)<C_x-d$.
\end{prop}
\begin{proof}
	The proof is the same as that of our Main lemma, \ref{mainlem}. We are given $x$ and have that no program shorter than $C_x$ can give $C(x)$. 
	We show that there is a $d$ such that by the time $(x,P_x^{*})$ has been enumerated in the dovetailing $D$, all pairs $(y,P_y)$ with $C(y)<C_x-d$ will have been enumerated, solving the function $C(y)$ for all such $y$. 
	\newline
	Say, to the contrary, that such a pair $(y,P_y)$ appears before $(x,P_x^{*})$. Then $P_y$ (with $l(P_y)< l(P_x^{*})$) computes $l(P_x^{*})$ as follows: Given $P_y$, run the dovetailing until $(y,P_y)$ appears. Now look for the shortest $P$ such that $(x,P)$ has appeared. This $P$ must be $P_x^{*}$. Obtain $C(x)$ as $l(P_x^{*})$, a contradiction.
\end{proof}
\begin{cor}
	There is a constant $c$, such that, if $C_x>d+c$ then $(x,C(x))$ solves the halting problem up to length $d$.
\end{cor}

\begin{cor}
	There is a constant $c $, such that, if $C_x>d+c$ then $(x,C(x))$ solves the Busy beaver problem up to length $d$.
\end{cor}


\begin{thebibliography}{99}
\bibitem{lv} Li, Ming, and P. M. Vit\'{a}nyi. 2019. An introduction to Kolmogorov complexity and its applications. Cham Springer Series:	Texts in Computer Science
\bibitem{Slaman}  A. Ku\v{c}era and T. A. Slaman. Randomness and recursive enumerability.
SIAM Journal on Computing, 31:199-211, 2001.
\bibitem{Shen} B. Bauwens, A. Shen, Complexity Of Complexity And Strings With Maximal Plain And Prefix Kolmogorov Complexity, The Journal of Symbolic Logic. 79 (2014) 620632. doi:10.1017/jsl.2014.15.
\bibitem{Chait} G.J. Chaitin. Information-theoretic characterizations of recursive infinite strings. Theor. Comput. Sci., 2:45\-48, 1976.
\bibitem{Down} R. G. Downey. D.R. Hirschfeldt. Algorithmic randomness and complexity, Springer Verlag 2010
\bibitem{Nies} A. Nies. Computability and Randomness. Oxford, England: Oxford University Press UK. (2008)
\bibitem{Csima} Csima, B. F., and Shore, R. A. (2007). The settling-time reducibility ordering. Journal of Symbolic Logic, 72(3), 1055-1071. doi:10.2178/jsl/1191333856
\bibitem{Gacs}P. G\'{a}cs, On the symmetry of algorithmic information, Soviet Math. Dokl., vol. 15 (1974), pp. 1477-1480.
\bibitem{Merkle Stephan} W. Merkle and F. Stephan. On C-degrees, H-degrees and T-degrees. In
Twenty-Second Annual IEEE Conference on Computational Complexity (CCC 2007). IEEE Computer Society Press, San Diego, CA, 2007.
\end{thebibliography}
\end{document}